\documentclass[final]{dmtcs-episciences}
\usepackage[utf8]{inputenc}
\usepackage{subfigure}

\usepackage{crossreftools}
\usepackage{amsmath}
\usepackage{amsfonts}
\usepackage{amssymb} 
\usepackage{mathtools}
\usepackage[dvipsnames]{xcolor}  
\definecolor{darkgreen}{RGB}{0,150,0}
\usepackage{tikz}  
\usepackage{pgf}  
\usepackage{xspace}

\usepackage{hyperref}   
\usepackage{stmaryrd}  

\usepackage[english]{babel} 
\usepackage{float}  % precise placement of figures 
\usepackage[T1]{fontenc}
\usepackage[utf8]{inputenc}
\usepackage{microtype} 
\usepackage{paralist} 
\usepackage{graphicx}
\usepackage{tikz}
\usetikzlibrary{arrows,positioning}   
\usepackage{centernot} % negation of \rightarrow
\usepackage{dsfont} 
\newcommand{\struct}[1]{\mathfrak{#1}} 

\newcommand{\Models}[1]{\llbracket{#1}\rrbracket^{<\omega}}

\newcommand{\ignore}[1]{}

\newcommand{\EL}{\ensuremath{\mathcal{EL}}\xspace}  

\usepackage{amsthm}
 
\theoremstyle{plain}

\newtheorem{theorem}{Theorem}
\newtheorem{proposition}{Proposition}
\newtheorem{lemma}{Lemma}
\newtheorem{corollary}{Corollary}
\newtheorem{definition}{Definition} 
\newtheorem{myclaim}{Claim}  
\newtheorem{observation}{Observation}

\theoremstyle{remark}
\newtheorem{example}{Example} 

\theoremstyle{definition}
\newtheorem{question}{Question} 

\usepackage{thmtools, thm-restate}

\newcommand{\bA}{\mathfrak{A}}
\newcommand{\bB}{\mathfrak{B}}
\newcommand{\bC}{\mathfrak{C}}

\usepackage{relsize}
\newcommand{\of}[1]{\textup{\relsize{-2}\raisebox{0.5pt}{[}\kern0.025em$#1$\kern0.025em\raisebox{0.5pt}{]}}}

\newcommand{\boldf}[1]{\bar{#1}}

\usepackage{microtype}
	
	\author{Manuel Bodirsky\affiliationmark{1}\thanks{Manuel Bodirsky has received funding from the European Research Council through the ERC Consolidator Grant 681988 (CSP-Infinity).}  \and Jakub Rydval\affiliationmark{2}\thanks{Jakub Rydval is supported by DFG GRK 1763 (QuantLA)}   \and Andr\'e Schrottenloher\affiliationmark{3}\thanks{Andr\'e Schrottenloher is supported by ERC-ADG-ALGSTRONGCRYPTO (project 740972).}}

\title[Universal Horn Sentences  
and the Joint Embedding Property]{Universal Horn Sentences and the Joint Embedding Property}
	
	\affiliation{
  Institute of Algebra, TU Dresden, Germany%, \texttt{manuel.bodirsky@tu-dresden.de} 
  \\
  Institute of Theoretical Computer Science, TU Dresden, Germany%, \texttt{jakub.rydval@tu-dresden.de}
  \\
 Cryptology Group, CWI, Amsterdam, The Netherlands%, \texttt{andre.schrottenloher@m4x.org}
}

\received{2021-05-02}
\revised{2021-11-19, 2022-03-09, 2022-04-08}
\accepted{2022-04-09}

\begin{document}

%\publicationdetails{VOL}{2015}{ISS}{NUM}{SUBM}

\publicationdetails{23}{2022}{2}{4}{7435}

	\maketitle            
\begin{abstract}The finite models of a universal sentence $\Phi$ in a finite relational signature are the age of a structure if and only if $\Phi$ has the \emph{joint embedding property}. 
	We prove that the computational problem whether a given universal sentence $\Phi$ has the joint embedding property is undecidable, even if $\Phi$ is additionally Horn and the signature of $\Phi$ only contains relation symbols of arity at most two. 
\end{abstract}

\section{Introduction} 
%
%
%
%The preservation theorem of McKinsey states that a universal sentence $\Phi$ is equivalent to a universal Horn sentence if and only if the class of finite models of $\Phi$ is closed under direct products~\cite{HodgesLong}. 
%
The \emph{age} of a relational structure $\struct{D}$ is a term from model-theory used for the class of all finite  structures that can be embedded to $\struct{D}$.
Universal Horn sentences which describe the age of a relational structure are particularly relevant in theoretical computer science; we mention two contexts where they appear. 

The \emph{constraint satisfaction problem} (CSP) of a structure $\struct{D}$ with a finite relational signature $\tau$ is the problem of deciding whether a given finite $\tau$-structure $\bA$ has a homomorphism to $\struct{D}$.
If the domain of $\struct{D}$ is finite, then the CSP of $\struct{D}$ is clearly in NP. Moreover, finite-domain CSPs admit a dichotomy between P and NP \cite{bulatov2017dichotomy,zhuk2017proof}.
In the infinite, the situation is much more complicated and the CSP might even be undecidable, e.g., the solvability of Diophantine equations can be represented as the CSP of $(\mathbb{Z};+,\cdot,\{0\},\{1\})$.
What one usually does when studying infinite-domain CSPs is to restrict the set of admissible templates $\struct{D}$ in order to guarantee a certain upper bound on complexity while still being able to model interesting decision problems.
One such restriction is requiring the age of $\struct{D}$ to be definable by a universal first-order sentence, we then say that $\struct{D}$ is \emph{finitely bounded}.
In this case, the CSP of $\struct{D}$ is in NP. For every satisfiable instance $\struct{A}$, we can guess a finite structure $\struct{B}$ which admits a surjective homomorphism from $\struct{A}$ and  embeds into $\struct{D}$.
To verify that $\struct{B}$ indeed embeds into $\struct{D}$, we must only confirm that it satisfies the fixed universal sentence, which can be done in time polynomial in the size of $\struct{A}$. 
If the age of $\struct{D}$ is definable by a universal Horn sentence, then the CSP of $\struct{D}$ is even guaranteed to be solvable in polynomial time by a Datalog program.
Here, we would simply saturate the relations of $\struct{A}$ under application of the fixed set of  Horn rules and reject if and only if we can derive the empty clause. 
A wide variety of decision problems are CSPs of reducts of finitely bounded structures, e.g., digraph acyclicity or satisfiability of a fixed connected MMSNP sentence \cite{Book}.
However, in general it is not clear which universal (Horn) sentences actually describe the age of a relational structure, and this question can be considered as an interesting decision problem if one intends to study the CSPs of reducts of finitely bounded structures.

\emph{Description logics} are knowledge representation languages frequently used to formalize
ontologies.
To enable reference to concrete objects and predefined predicates on these objects, a description logic can be extended by a concrete domain, which is usually represented by a fixed relational structure. 
By a result in~\cite{PReport}, if the age of a relational structure $\struct{D}$ is defined by a universal sentence $\Phi$, then $\struct{D}$ yields a tractable concrete domain extension of the description logic $\EL$  if and only if $\Phi$ is equivalent to a universal Horn sentence. 
To enable user-definability of concrete domains for $\EL$, one would need an algorithm accepting precisely those universal Horn sentences which describe the age of a relational structure.

We show that the question whether a given universal sentence is equivalent to a universal Horn sentence can be decided effectively. 
This motivates the question whether one can also decide which universal Horn sentences describe the age of a relational structure.
By a classical result of Fra\"{i}ss\'{e}, the finite models of a universal sentence $\Phi$ in a finite relational signature are the \emph{age} of some structure  if and only if $\Phi$ has the \emph{joint embedding property}~\cite{HodgesLong}. 
The problem whether a given universal sentence $\Phi$ has the joint embedding property is undecidable by a result of~\cite{BraunfeldUndec}, but the theories used by Braunfeld to prove the undecidability are not universal Horn. 
We prove that the joint embedding property remains undecidable even if the given universal sentence is additionally Horn, and if the signature is binary. Our proof is based on the undecidability of universality of context-free languages and shorter than the argument of Braunfeld.

\section{Preliminaries} 
The set $\{1,\dots,n\}$ is denoted by $[n]$.   
We often use the bar notation for tuples; for a tuple $\boldf{a}$ indexed by a set $I$, 
the value of $\boldf{a}$ at the position $i\in I$ is denoted by  $\boldf{a}\of{i}$.
For a function $f\colon A \rightarrow B$ and a tuple $\boldf{a}\in A^k$, we set $f(\boldf{a})\coloneqq \big( f(\boldf{a}\of{1}),\dots, f(\boldf{a}\of{k}) \big)$.
We always use capital Fraktur letters $\bA$, $\bB$, $\bC$, etc to denote structures, and the corresponding capital Roman letters $A$, $B$, $C$, etc for their domains.

Let $\tau$ be a finite relational signature.
%
% An \emph{atomic $\tau$-formula}, or \emph{$\tau$-atom} for short, is a formula of the form $x=y$ or of the form $R(x_1,\dots,x_n)$ for $R \in \tau$ and (not necessarily distinct) variables $x_1,\dots,x_n$. 
%
An \emph{atomic $\tau$-formula}, or \emph{$\tau$-atom} for short, is a formula of the form $R(x_1,\dots,x_n)$ for $R \in \tau$ and (not necessarily distinct) variables $x_1,\dots,x_n$.
Note that in this text equality atoms, i.e., expressions of the form $x=y$, are not permitted.
A class $\mathcal{C}$ of relational $\tau$-structures has the \emph{joint embedding property (JEP)} if 
for every pair $\struct{B}_1,\struct{B}_2 \in \mathcal{C}$, there exists $\struct{C}\in \mathcal{C}$ and   embeddings $f_{i}\colon \struct{B}_i \hookrightarrow \struct{C}$ ($i\in [2]$).
It is closed under forming \emph{disjoint unions} if  $\struct{C} \in \mathcal{C}$ and $f_{i}\colon \struct{B}_i \hookrightarrow \struct{C}$ can be chosen so that $f_1(B_1)\cup f_2(B_2)=C$, $f_1(B_1)\cap f_2(B_2) = \emptyset$, and $R^{\struct{C}} = f_1(R^{\struct{B}_1})\cup f_2(R^{\struct{B}_2})$ for every $R\in \tau$.

If $\bA$ is a $\tau$-structure and $\Phi$ is a $\tau$-sentence, we write $\bA \models \Phi$ if $\bA$ satisfies $\Phi$. If $\Phi$ and $\Psi$ are $\tau$-sentences, we write $\Phi \models \Psi$
if every $\tau$-structure that satisfies $\Phi$ also satisfies $\Psi$. 
For universal sentences, this is equivalent
to the statement that every \emph{finite}
$\tau$-structure that satisfies $\Phi$ also satisfies $\Psi$, by a standard application of the  compactness theorem of first-order logic. 
If $\Phi$ is a $\tau$-sentence, we denote the class of all finite models of $\Phi$ by $\Models{\Phi}$.
Clearly, every universal formula can be written as
$\forall x_1,\dots,x_n\ldotp  \psi$ where $\psi$ is quantifier-free and in conjunctive normal form. 
Conjuncts of $\psi$ are also called \emph{clauses}; clauses are disjunctions of atomic formulas and negated atomic formulas.
A clause $\phi$ is a \emph{weakening} of a clause $\psi$ if every disjunct of $\psi$ is also a disjunct of $\phi$.
We sometimes omit universal quantifiers in universal sentences (all first-order variables are then implicitly universally quantified). 
It is well-known that a class $\mathcal{C}$ of finite $\tau$-structures can be described by forbidding isomorphic copies of fixed finitely many finite $\tau$-structures if and only if $\mathcal{C} = \Models{\Phi}$ for some universal $\tau$-sentence $\Phi$.
With this observation, it is easy to see that Braunfeld~\cite{BraunfeldUndec} proved undecidability for the 
following decision problem. His proof is based on a reduction from the unrestricted tiling problem.  

\par\medskip  \noindent 
INPUT: A universal sentence $\Phi$ in a finite relational signature. \\
QUESTION: Does $\Models{\Phi}$ have the JEP?
\medskip 

%We assume that equality $=$ is always available when building formulas; thus, atomic formulas are of the form $x_{i}=x_{j}$ for $i,j \in \{1,\dots,n\}$, or $R(y_{i_1},\dots ,y_{i_k})$ for some relation symbol $R$ of arity $k$ and $i_1,\dots,i_k \in \{1,\dots,n\}$. 
%
A universal sentence $\Phi$ is called \emph{Horn} if its
quantifier-free part is a conjunction of \emph{Horn clauses}, i.e., disjunctions of positive or negative atomic formulas each having at most one \emph{positive} disjunct. Every Horn clause can be written equivalently as an implication  
$$\phi_1 \wedge \cdots \wedge \phi_n \Rightarrow \phi_0$$
where $\phi_1,\dots,\phi_n$ are atomic formulas (possibly $n=0$), and $\phi_0$ is either an atomic formula or $\bot$ which stands for the empty disjunction.
We refer to $\phi_1 \wedge \cdots \wedge \phi_n$ as the \emph{premise}, and to $\phi_0$ as the \emph{conclusion} of the Horn clause. 
A Horn clause over a signature $\tau$ is called a \emph{tautology} if it holds in all $\tau$-structures, i.e., if the conclusion equals one of the conjuncts of the premise.

We say that a Horn clause $\phi_1 \wedge \cdots \wedge \phi_n \Rightarrow \phi_0$ \emph{can be applied} to a structure $\struct{A}$ if $\phi_1 \wedge \cdots \wedge \phi_n \wedge \neg \phi_0$ is satisfiable in $\struct{A}$.
%  
	%We say that a $\tau$-structure $\struct{A}'$ is \emph{obtained from $\struct{A}$ by applying $\alpha$} if 
	  %$\phi_0$ is of the form $R(x_{i_1},\dots,x_{i_k})$ where $x_1,\dots,x_m$ are all the variables of $\alpha$, $A=A'$, and there exists a tuple $\boldf{a} \in A^m$ such that 
	%$\struct{A}\models (\phi_1 \wedge \cdots \wedge \phi_n)(\boldf{a})$, $R^{{\struct A}'} = R^{\struct{A}} \cup \{(\boldf{a}\of{i_1},\dots,\boldf{a}\of{i_k}) \}$, and $S^{{\struct A}'} = S^{\struct{A}}$ for all $S \in \tau \setminus \{R\}$. 
	%	
	Let $\Phi$ be a Horn sentence. We say that $\struct{A}$ is \emph{closed under application of Horn clauses from $\Phi$} if no conjunct of $\Phi$ is applicable to $\struct{A}$, i.e., if $\struct{A}$ is a model of $\Phi$.

\begin{definition} Let $\Phi$ be a universal Horn sentence and $\psi$ a Horn clause, both in a fixed relational signature $\tau$. An \emph{SLD-derivation of $\psi$ from $\Phi$ of length $k$} is a finite sequence of Horn clauses $\psi_0,\dots, \psi_k = \psi$ such that $\psi_0$ is a conjunct in $\Phi$ and each $\psi_i$ ($1\leq i \leq k$) is a \emph{(binary) resolvent} of $\psi_{i-1}$ and a conjunct $\phi_{i}$ from $\Phi$, i.e., for some atomic formula $\psi^{j}_{i-1}$, we have
	\begin{center} 
		\begin{tabular}{c} $\overset{\psi_{i-1}}{\overbrace{ \psi_{i-1}^1 \wedge  \cdots \wedge     \psi_{i-1}^j\wedge  \cdots \wedge   \psi_{i-1}^{n_{i-1}}  \Rightarrow \psi_{i-1}^0}}   \qquad   \overset{\phi_{i}}{\overbrace{ \phi_{i}^1 \wedge  \cdots \wedge    \phi_{i}^{m_{i}} \Rightarrow \psi_{i-1}^j}}$ \\[0.5em]  \hline \\[-0.5em] 
			$\underset{\psi_i}{\underbrace{\psi_{i-1}^1 \wedge  \cdots \wedge   \psi_{i-1}^{j-1} \wedge   \phi_{i}^1  \wedge  \cdots \wedge   \phi_{i}^{m_{i}}\wedge     \psi_{i-1}^{j+1}\wedge \cdots \wedge   \psi_{i-1}^{n_{i-1}} \Rightarrow \psi_{i-1}^0 }}	$
		\end{tabular} 
	\end{center} 
	There exists an \emph{SLD-deduction} of $\psi$ from $\Phi$, written as $\Phi \vdash \psi$, if $\psi$ is a tautology  or a weakening of  a Horn clause that has an SLD-derivation from $\Phi$ up to 
	%a substitution of variables.   
	renaming variables.
\end{definition}

The following theorem presents a fundamental property of universal Horn sentences: that SLD-deduction is a sound and complete calculus for entailment of Horn clauses by Horn sentences.
% \red{I don't think the statement is correct if we allow equality atoms. Probably the best solution is to additionally allow the application of some Horn clauses that code how equality behaves.}

\begin{theorem}[Theorem~7.10 in \cite{NienhuysW97}]  \label{SLD-deduction}Let $\Phi$ be a universal Horn sentence and $\psi$ a Horn clause, both in a fixed signature $\tau$. Then \[\Phi \models \psi \quad \text{iff}\quad \Phi \vdash  \psi. \]
\end{theorem}

\section{Universal Horn Sentences and the JEP}

This section deals with two separate topics.
First, we discuss the semantical difference between universal sentences and universal Horn sentences which is described by the convexity condition.
Second, we present some important observations about the JEP in the context of universal Horn sentences which we later use in the proof of our main result.
Let $\tau$ be a relational signature. 

A universal $\tau$-sentence $\Phi$ is called \emph{convex} if for all atomic $\tau$-formulas $\phi_1,\dots,\phi_n$, $\psi_1,\dots,\psi_k$, if
$$\Phi \models \bigwedge_{i \in [n]} \phi_i \Rightarrow \bigvee_{j \in [k]} \psi_j$$
then there exists $j \in [k]$ such that
$$\Phi \models \bigwedge_{i \in [n]} \phi_i \Rightarrow \psi_j.$$
We say that $\Phi$ is \emph{preserved in products}
if for every non-empty family $(\bA_i)_{i \in I}$ of models of $\Phi$, the direct product $\prod_I \bA_i$ is also a model of $\Phi$. 
The following is well known; e.g., the direction  \eqref{ctr3}$\Rightarrow$\eqref{ctr1} is Corollary~9.1.7 in~\cite{HodgesLong}. Since we also need \eqref{ctr4} we provide a proof for the convenience of the reader. 

\begin{theorem}[McKinsey]\label{convexity}
	Let $\Phi$ be a universal sentence with relational signature $\tau$. Then the following are equivalent. 
	\begin{enumerate}
		\item \label{ctr1} $\Phi$ is convex; 
		\item \label{ctr2} $\Phi$ is equivalent to a universal Horn sentence; 
		\item \label{ctr3} $\Phi$ is preserved in products;
		\item \label{ctr4} $\Phi$ is preserved in binary products of finite structures. 
	\end{enumerate}
	%\begin{enumerate} 
	%\item $\Phi$ is preserved in arbitrary products;
	%\item $\Phi$ is preserved in binary products;
	%\item $\Phi$ is convex. 
	%\end{enumerate}
\end{theorem}
\begin{proof} ``\ref{ctr1}$\Rightarrow$\ref{ctr2}'': We may assume that $\Phi$ is in prenex normal form and that its quantifier-free part $\phi$ is in conjunctive normal form.
	Every conjunct in $\phi$ is equivalent to an implication of the form $\bigwedge_{i \in [n]} \phi_i \Rightarrow \bigvee_{j \in [k]} \psi_j.$
	Since $\Phi$ is convex, we can replace the conjunct by $ \bigwedge_{i \in [n]} \phi_i \Rightarrow \psi_j$ for some $j\in [k]$.
	In this way, $\Phi$ can be rewritten into an equivalent universal Horn sentence.

	``\ref{ctr2}$\Rightarrow$\ref{ctr3}'':  Corollary~9.1.6 
	in~\cite{HodgesLong}.  
	
	``\ref{ctr3}$\Rightarrow$\ref{ctr4}'': This direction is trivial.
	
	``\ref{ctr4}$\Rightarrow$\ref{ctr1}'': Suppose that $\Phi$  has $m$ variables  and is not convex, i.e., $\Phi \models \bigwedge_{i \in [n]} \phi_i \Rightarrow \bigvee_{j \in [k]} \psi_j$ but, for every $j \in [k]$,  there exists a model $\struct{A}_j$ of $\Phi$
	such that $\struct{A}_j \models  (\bigwedge_{i \in [n]} \phi_i \wedge \neg \psi_j)(\boldf{t}_j)$ for some tuple $\boldf{t}_{j} \in A_j^m$.
	We may assume that each $\struct{A}_{j}$ is finite; otherwise we replace it with its substructure on the coordinates of $\boldf{t}_j$ while preserving the desired properties.
	For $\boldf{s}_{j} \coloneqq ((\boldf{t}_1\of{1},\dots,\boldf{t}_j\of{1}),\dots,(\boldf{t}_1\of{m},\dots,\boldf{t}_j\of{m}))$ we have 
	$$\prod_{i \in [j]} \struct{A}_{i} \models   \Big(\bigwedge_{i \in [n]} \phi_i \wedge \bigwedge _{i \in [j]} \neg \psi_{i}\Big)(\boldf{s}_{j}).$$
	It follows by induction on $j \in [k]$ that if $\Phi$ is preserved in binary products of finite structures, then $\prod_{i \in [j]} \struct{A}_{i} \models \Phi$. We then obtain a contradiction for  $j=k$.  
\end{proof}

It is natural to consider the following question as a decision problem.

\par\medskip  \noindent 
INPUT: A universal sentence $\Phi$ in a finite relational signature. \\
QUESTION: Is $\Phi$ equivalent to a universal Horn sentence?
\medskip 

%\noindent  In the context of the results of \cite{PReport}, one might also be interested in the variant where the input is restricted to those $\Phi$ for which  $\Models{\Phi}$ has the JEP. 

\begin{proposition} \label{prop:equivalence_horn}
	Deciding whether a given universal sentence $\Phi$ is equivalent to a universal Horn sentence is   $\Pi^p_2$-complete. The problem is $\Pi^p_2$-hard even when the signature is limited to unary relation symbols. 
\end{proposition}
\begin{proof}  We first prove containment in $\Pi^p_2$. For a given universal sentence $\Phi$, let  $\phi(x_1,\dots,x_n)$ be the quantifier-free part of $\Phi$. 
	If $\Phi$ is not equivalent to a universal Horn sentence, then, by Theorem~\ref{convexity}, $\Phi$ has two finite models $\struct{A},\struct{B}$ such that $\struct{A}\times \struct{B}$ is not a model of $\Phi$.
	This means that there exists $\boldf{t}\in (A\times B)^n$ such that $\struct{A}\times \struct{B} \centernot{\models} \phi(\boldf{t})$.
	But then, by the definition of product of structures, there exist substructures $\struct{A}'$ of $\struct{A}$ and $\struct{B}'$ of $\struct{B}$ of size at most $n$ with $\boldf{t}\in (A'\times B')^n$ and $\struct{A}'\times \struct{B}' \centernot{\models} \phi(\boldf{t})$.
	Since models of $\Phi$ are preserved under taking substructures, we have $\struct{A}' \models \Phi$ and $\struct{B}' \models \Phi$.
	Conversely, if there exist two models of $\Phi$ of size at most $n$ whose product is not a model of $\Phi$, then clearly $\Phi$ is not equivalent to a universal Horn sentence by Theorem~\ref{convexity}.

	The argument above shows that the following algorithm is sound and complete for the complement of the original problem.
	We first guess two models $\struct{A}$ and $\struct{B}$ of $\Phi$ of size at most $n$ such that $\struct{A}\times \struct{B}$ is not a model of $\Phi$.
	The latter can be verified in time polynomial in $n$ by guessing a tuple $\boldf{t}\in (A\times B)^n$ such that $\struct{A} \times \struct{B}\centernot{\models} \phi(\boldf{t})$.
	Verifying $\struct{A}\models \Phi$ and $\struct{B}\models \Phi$ iteratively would require a loop over up to $n^n$ many tuples, which would not yield an efficient procedure. Instead, we deal with the verification  using a coNP oracle that guesses any potential tuple witnessing that $\struct{B}\centernot{\models} \Phi $ or $\struct{A}\centernot{\models} \Phi$.
	
	Since this shows that the complement of the original problem is in $\Sigma^p_2$, 
	the original problem itself is in $\Pi^p_2$.

	The $\Pi^p_2$-hardness can be shown by a reduction from the complement of the propositional $\exists\forall$SAT problem.
	Consider an instance  
	\begin{align} 
		\exists X_1,\dots, X_{k}\forall X_{k+1},\dots, X_{\ell}\ldotp  \Psi(X_1,\dots, X_{\ell}) \label{eq:exists_forall_instance}
	\end{align}
	of propositional $\exists\forall$SAT.
	We first obtain the signature  $\tau = \{C_{1},\dots, C_{k},C,L,R\}$ consisting of unary symbols only.
	Let $\psi(x,x_{k+1},\dots,x_{\ell} )$ be the quantifier-free $\tau$-sentence obtained from  $\Psi(X_1,\dots, X_{\ell})$ by replacing each propositional variable $X_i$ with $C_i(x)$ if $i\in [k]$, and with $C(x_i)$ if $i\in [\ell]\setminus [k]$.
	Now we set  
	\begin{align*}
		\Phi  \coloneqq \forall x,y,x_{k+1},\dots,x_{\ell} \big( \neg C(x) \wedge C(y) \big) \Rightarrow \Big(\psi(x,x_{k+1},\dots,x_{\ell} )  \wedge \big(L(x) \vee R(x)\big)\Big) .
	\end{align*} 
	The idea is to show that: $\Phi$ is equivalent to $\forall x,y \big( C(y) \Rightarrow C(x)\big)$ if \eqref{eq:exists_forall_instance} is not satisfiable, and otherwise $\Phi$ is not  equivalent to any universal Horn sentence.
	
	``$\Rightarrow$'':  Suppose that \eqref{eq:exists_forall_instance} is satisfiable, i.e., there exists a map $f\colon \{X_1,\dots, X_{k}\} \rightarrow \{0,1\}$ such that every map $f' \colon\{X_1\dots, X_{\ell}\} \rightarrow \{0,1\}$ which extends $f$ is a satisfactory assignment for $\Psi(X_1,\dots, X_{\ell})$. 
	Let $\struct{A}_L$ be the $\tau$-structure over $\{a_1,a_2\}$ such that 
	\begin{enumerate}
		\item \label{it:first_part} for every $i\in [k]$ and $j\in [2]$,  $\struct{A}_L \models C_i(a_j)$ if and only if $f(X_i)=1$,
		\item \label{it:second_part} $\struct{A}_L \models \neg C(a_1)\wedge L(a_1)$, and $\struct{A}_L \models  C(a_2)$.
	\end{enumerate}
	We define $\struct{A}_R$ analogously by switching the roles of $L$ and $R$ in item~\ref{it:second_part} above.
	It follows directly from our assumption about $f$, item~\ref{it:first_part}, and item~\ref{it:second_part} that $\struct{A}_L \models \Phi$. 
	We also clearly get $\struct{A}_R \models \Phi$ since the construction of $\struct{A}_L$ and $\struct{A}_R$ is symmetrical w.r.t.\ $\Phi$.
	Now consider the structure $\struct{A}_L\times \struct{A}_R$.
	%
	%By item~\ref{it:first_part}, for every $\boldf{a} \in \{a_1,a_2\}^2$ and every $i\in [k]$, 
	% 
	%$\struct{A}_L\times \struct{A}_R\models  C_i(\boldf{a})$ if and only if $f(X_i)=1$.
	%
	%By our assumption about $f$, we have that  $\struct{A}_L\times \struct{A}_R\models  \psi((a_1,a_1),(a_1,a_1),\dots,(a_1,a_1))$.
	%
	We  have that $\struct{A}_L\times \struct{A}_R\models \neg C((a_1,a_1)) \wedge C((a_2,a_2))$.
	However, $\struct{A}_L\times \struct{A}_R\centernot{\models} L((a_1,a_1)) \vee R((a_1,a_1))$.
	Thus, $\Models{\Phi}$ is not preserved under products, which means that $\Phi$ is not equivalent to any universal Horn sentence by Theorem~\ref{convexity}.
	
	``$\Leftarrow$'': Suppose that $\Phi$ is not equivalent to any universal Horn sentence. By Theorem~\ref{convexity}, $\Models{\Phi}$ is not preserved in binary products of finite structures, i.e., there exist $\struct{A},\struct{B}\in \Models{\Phi}$ such that $\struct{A} \times \struct{B} \centernot{\models} \Phi$.
	
	First, suppose that either $\struct{A}\models C(a)$ for every $a\in A$ or $\struct{A}\models \neg C(a)$ for every $a\in A$, and either $\struct{B}\models C(b)$ for every $b\in B$ or $\struct{B}\models \neg C(b)$ for every $b\in B$.
	Then it also holds that either $\struct{A}\times \struct{B} \models C((a,b))$ for every $(a,b)\in A\times B$ or $\struct{A}\times \struct{B} \models \neg C((a,b))$ for every $(a,b)\in A\times B$.
	But then we clearly have  $\struct{A}\times \struct{B} \models \Phi$, a contradiction to our original assumption.

	Next suppose that, without loss of generality, there exist $a_1,a_2 \in A$ such that $\struct{A}\models \neg C(a_1)\wedge C(a_2)$.
	Let $f \colon \{X_1,\dots, X_{k}\} \rightarrow \{0,1\}$ be the map defined by $f(X_i) = 1$ if and only if $\struct{A}\models C_i(a_1)$.
	Let $f' \colon \{X_1,\dots, X_{\ell}\} \rightarrow \{0,1\}$ be an arbitrary extension of $f$.
	We want to show that $f'$ is a satisfactory assignment for $\Psi(X_1,\dots, X_{\ell})$. 
	For this, we consider the map $f''\colon \{x_{k+1},\dots, x_{\ell} \} \rightarrow \{a_1,a_2\}$ given by $f''(x_i) = a_2$ if and only if $f'(X_i) = 1$.
	Since $\struct{A} \models \Phi_1$, we have 
	$\struct{A} \models \psi(a_1,f(x_{k+1}),\dots,f(x_{\ell})).$
	By the definition of $\psi$ and $f''$, the map $f'$ is a satisfying assignment for $\Psi(X_1,\dots, X_{\ell})$. 
	We conclude that \eqref{eq:exists_forall_instance} is satisfiable.  
\end{proof} 

In the context of the results in \cite{PReport}, the $\Pi^p_2$-hardness in Proposition~\ref{prop:equivalence_horn} would also be interesting under the assumption that $\Models{\Phi}$ has the JEP. 
This can be achieved with a simple trick where we introduce a fresh binary symbol $E$ into $\tau$ and expand the premise of the implication in $\Phi$ by the conjunction $E(x,y) \wedge E(y,x_{k+1})\wedge\bigwedge_{i=k+1}^{\ell-1} E(x_i,x_{i+1})$.
As a consequence, $\Models{\Phi}$ is even closed under the formation of disjoint unions, which is a strong form of the JEP.
\begin{corollary}  
	Deciding whether a given universal sentence $\Phi$ such that $\Models{\Phi}$ has the JEP is equivalent to a universal Horn sentence is $\Pi^p_2$-hard, even when the signature is limited to binary relation symbols. 
\end{corollary}
The trick above is based on a general observation about connected clauses that we later use in our undecidability proof, in the form of Proposition~\ref{connected}.

\begin{definition} 
	A Horn clause $\phi \Rightarrow \psi$  is called \emph{connected} if the graph with vertex set $\{x_1,\dots,x_n\}$ where $x_i$ and $x_j$ form an edge if they appear jointly in a conjunct of $\phi$, forms a connected graph on all variables occurring in $\phi \Rightarrow \psi$ in the usual graph-theoretic sense.  
\end{definition}

% In the following, we fix a finite relational signature $\tau$. 
%
% Let $\struct{A},\struct{B}_1,\struct{B}_2$ be relational $\tau$-structures with $\struct{A}\hookrightarrow \struct{B}_i$ for $i\in [2]$.
%
% W.l.o.g.\ we may assume that $B_1 \cap B_2 = A$, otherwise we rename elements accordingly.
%  
% Then the \emph{free amalgam} of $(\struct{A},\struct{B}_1,\struct{B}_2)$ is defined as the union of $\struct{A}$, $\struct{B}_1$, and $\struct{B}_2$.
%
\begin{proposition} \label{connected} 
	Let $\Phi$ be a universal Horn sentence such that each conjunct in the quantifier-free part of $\Phi$ is connected. Then $\Models{\Phi}$ is closed under the formation of disjoint unions and therefore has the JEP.
\end{proposition}   
\begin{example} The class of all finite strict partial orders, defined by  
	\[ 	\forall x,y,z \big( x<y \wedge y<z \Rightarrow x<z  \big) \wedge \big( x<x \Rightarrow \bot   \big), \]
	is the age of the set of all finite subsets of $\mathbb{N}$ partially ordered by  set inclusion.
\end{example}
In fact, universal sentences preserved under disjoint unions can be fully characterized in terms of being equivalent to a universally quantified conjunction of connected clauses, see, e.g., Theorem~4.4 in \cite{compton1983some}.
One should not expect any similar normal form for universal sentences whose finite models have the JEP since this property is undecidable.
However, there is a useful reformulation of the JEP  for classes of finite structures defined by universal Horn sentences that greatly simplifies the presentation of our undecidability proof.
\begin{definition} 
	\label{def:domination} 
	Let 	$\Phi$ be a universal Horn sentence over the relational signature $\tau$, and let $\phi(\boldf{x})$ and $\psi(\boldf{x},\boldf{y})$ be conjunctions of atomic $\tau$-formulas.
	We write  %say that \emph{\blue{$\phi(\boldf{x})$  dominates $ \psi(\boldf{x},\boldf{y})$ modulo $\Phi$}} and write 
	 $\psi(\boldf{x},\boldf{y}) \leq_{\Phi} \phi(\boldf{x})$  if for every $\chi $ that is $\bot$ or an atomic $\tau$-formula with free variables among $\boldf{x}$
	$$
	\Phi \models \forall  \boldf{x},\boldf{y}\big(  \psi(\boldf{x},\boldf{y})  \Rightarrow \chi(\boldf{x}) \big) \quad \text{ implies }\quad  \Phi\models  \forall  \boldf{x} \big(\phi(\boldf{x})\Rightarrow \chi(\boldf{x}) \big).  
	$$
\end{definition}
%

%Note that \blue{domination} can be tested using SLD-deduction (Theorem~\ref{SLD-deduction}). 
%
%\begin{observation} Let $\Phi_1,\Phi_2$ be universal Horn sentences and $\psi$ a Horn clause, all over a common finite relational signature, such that $\Phi_1\wedge \Phi_2 \models \psi$. If there exists a symbol $R$ with the property that
%	\begin{itemize}
%		\item $R$ appears in the premise of every Horn clause in $\Phi_2$,
%		\item $R$ does not appear in $\psi$ at all, and
%		\item $R$ does not appear in the conclusion of any Horn clause in $\Phi_1$,
%	\end{itemize}
%then $\Phi_1 \models \psi$.
%\end{observation}
%%
%\begin{proof}
%	
%\end{proof} 
%
In our proofs later it will be useful to take a proof-theoretic perspective on the JEP.

\begin{lemma} \label{charact_jep_horn} Let $\Phi$  be a universal Horn 
	sentence over the relational signature $\tau$. Then the following are equivalent:
	\begin{enumerate} 
		\item \label{eq:lemma1_item1}  $\Models{\Phi}$ has the joint embedding property.
		\item  \label{eq:lemma1_item2} Suppose that   
		$\phi_1(\boldf{x}_1)$
		and $\phi_2(\boldf{x}_2)$ are  
		conjunctions of atomic formulas with disjoint sets of variables and such that $\phi_i(\boldf{x}_i) \wedge \Phi$ is satisfiable for both $i\in [2]$.
		Then 
		$
		\phi_1(\boldf{x}_1) \wedge \phi_2(\boldf{x}_2)  \leq_{\Phi} \phi_1(\boldf{x}_1).
		$
	\end{enumerate}
\end{lemma}  
\begin{proof}    
	``\ref{eq:lemma1_item1}$\Rightarrow$\ref{eq:lemma1_item2}'':  Let  $\phi_1(\boldf{x}_1)$ and $\phi_2(\boldf{x}_2)$ be as in the first part of item~\ref{eq:lemma1_item2}. 
	Let  $\Psi_1(\boldf{x}_1)$ and $\Psi_2(\boldf{x}_2)$ be the conjunctions of all atomic formulas implied by  $\Phi\wedge  \phi_1(\boldf{x}_1)$ and $\Phi \wedge  \phi_2(\boldf{x}_2)$, respectively. 
	By our assumption,  $\Phi \wedge \Psi_1(\boldf{x}_1)$ and $\Phi \wedge \Psi_2(\boldf{x}_2)$ are both  satisfiable.
	Define   $\struct{B}_1$ and $\struct{B}_2$ as the structures  whose domains consist of the 
	variables  $\{ \boldf{x}_1\of{1},\dots\}$, and $\{ \boldf{x}_2\of{1},\dots\}$, respectively, and 
	where $\boldf{z}$ is a tuple of a relation for $R\in \tau$ if the conjunct $R(\boldf{z})$ 
	is contained in $ \Psi_1$ or $\Psi_2$, respectively. 
	Note that, by construction,  $\struct{B}_1$ and $ \struct{B}_2$ satisfy every Horn clause in $\Phi$. Since $\Phi$ is universal Horn, this implies that $ \struct{B}_1, \struct{B}_2 \in \Models{\Phi}$.
	Since $\Models{\Phi}$ has the joint embedding property, there exists $\struct{C} \in \Models{\Phi}$ together with embeddings $f_i\colon \struct{B}_i \hookrightarrow \struct{C}$ for $i\in \{1,2\}$.
	By the construction of $\struct{B}_1$ and $\struct{B}_2$, it follows that  
	$\Phi \centernot{\models} \forall \boldf{x}_1,\boldf{x}_2\big(   \phi_1(\boldf{x}_1) \wedge \phi_2(\boldf{x}_2)  \Rightarrow \bot \big).$
	Let $\chi(\boldf{x}_1)$ be an atomic $\tau$-formula   such that 
	$
	\Phi \models \forall \boldf{x}_1,\boldf{x}_2\big(  \phi_1(\boldf{x}_1) \wedge \phi_2(\boldf{x}_2)  \Rightarrow \chi(\boldf{x}_1)  \big).
	$
	By the construction of  $\struct{B}_1$  and $\struct{B}_2$, and because  $f_1$ and $f_2$ are homomorphisms, there exist a tuple $\boldf{z}$ over $B_1$ such that 
	$\struct{C} \models \chi (f_1(\boldf{z}))$.
	Since $f_1$ is an embedding, we must also have $\struct{B}_1 \models \chi (\boldf{z})$.
	Thus, by the construction of  $\struct{B}_1$  and $\struct{B}_2$, it follows 
	that 
	$\Phi \models \forall \boldf{x}_1 \big( \phi_1(\boldf{x}_1) \Rightarrow \chi(\boldf{x}_1)  \big).$

	``\ref{eq:lemma1_item2}$\Rightarrow$\ref{eq:lemma1_item1}'':  Let $\struct{B}_1,\struct{B}_2 \in \Models{\Phi}$ be arbitrary.
	We construct a structure $\struct{C} \in \Models{\Phi}$ with $f_{i}\colon \struct{B}_i \hookrightarrow \struct{C}$ as follows.
	Without loss of generality we may assume that  $B_1\cap B_2 = \emptyset$.
	%
	%Let $\boldf{x}$ be a tuple of variables representing the elements of $B_1\cap B_2$ in some order, and let
	%
	Let $ \phi_1(\boldf{x}_1)$ and $ \phi_2(\boldf{x}_2)$ be the conjunctions of all atomic $\tau$-formulas which hold in $\struct{B}_1$ and $\struct{B}_2$, respectively.  
	By construction,  $ \phi_i(\boldf{x}_i) \wedge \Phi$ is satisfiable for both $i\in [2]$.
	Let $\Psi(\boldf{x}_1,\boldf{x}_2)$ be the conjunction of all atomic formulas implied by $\Phi\wedge \phi_1(\boldf{x}_1) \wedge \phi_2(\boldf{x}_2)$.
	We claim that $\Phi \wedge \Psi$ is satisfiable:
	otherwise, $\Phi \models \forall \boldf{x}_1,\boldf{x}_2 ( \phi_1(\boldf{x}_1) \wedge \phi_2(\boldf{x}_2)\Rightarrow \bot)$,
	and then item~\ref{eq:lemma1_item2} implies that  $\Phi \models \forall \boldf{x}_1 ( \phi_1(\boldf{x}_1) \Rightarrow \bot)$, which is impossible since  $\struct{B}_1 \models \Phi$.
	Define $\bC$ as the structure with domain $\{\boldf{x}_1\of{1},\dots, \boldf{x}_2\of{1},\dots\}$ and such that $R^{\struct{C}}$ contains a tuple $\boldf{z}$  if and only if   $\Psi$ contains the conjunct $R(\boldf{z})$. 
	For $i\in [2]$, let $f_i$ be the identity map.  We claim that $f_i$ is an embedding  from $\struct{B}_i$ to $\struct{C}$.
	It is clear from the construction of $\bC$ that $f_i$ is a homomorphism. 
	Suppose for contradiction that there exists $ R\in \tau$ and a tuple $\boldf{z}$ over $ B_i$ such that 
	$\boldf{z} \notin R^{\struct{B}}_i$ while   $f_i(\boldf{z}) \in R^{\struct{C}}$.
	For the sake of notation, we assume that $i = 1$; the case that $i=2$ can be shown analogously. 
	Note that the construction of $\bC$ implies that  $\Phi \models \forall \boldf{x}_1,\boldf{x}_2  \big(  \phi_1(\boldf{x}_1) \wedge  \phi_2(\boldf{x}_2)  \Rightarrow R(\boldf{z})\big)$.
	Then item~\ref{eq:lemma1_item2} implies that 
	$\Phi \models \forall \boldf{x}_1 \big( \phi_1(\boldf{x}_1) \Rightarrow R(\boldf{z})\big)$, a contradiction to $\struct{B}_1\in \Models{\Phi}$.
	Thus,  $f_i$ is an embedding from $\struct{B}_i$ to $\struct{C}$.
	This concludes the proof of the joint embedding property.   
\end{proof} 
We mention in passing that \eqref{eq:lemma1_item1} and \eqref{eq:lemma1_item2}~are equivalent even if equality is permitted in $\Phi$.

\section{Undecidability of the JEP}
In this section we prove that the problem of deciding whether the class of all finite models of a given universal Horn sentence $\Phi$ has the joint embedding property is undecidable. Our proof is based on a reduction from the problem of deciding the universality of a given context-free grammar. 
As usual, the \emph{Kleene plus} of $\Sigma$, denoted by $\Sigma^{+}$, is the set of all finite words over $\Sigma$ of length $\geq 1$.
A context-free grammar (CFG) is a $4$-tuple $G=(N,\Sigma,P,S)$ where 
\begin{itemize}
	\item $N$ is a finite set of \emph{non-terminal symbols},
	\item $\Sigma$ is a finite set of \emph{terminal symbols},
	\item $P$ is a finite set of \emph{production rules} of the form $A \rightarrow w$ where $A\in N$ and $w \in (N \cup \Sigma)^{+}$,
	\item $S \in N$ is the \emph{start symbol}.
\end{itemize}  
For $u,v\in (N\cup \Sigma)^{+}$ we write $u \rightarrow_{G} v $ if there are $x,y\in(N \cup \Sigma)^{+} $ and $(A \rightarrow w)\in P$ such that $u=xAy$ and $v=xwy$.
The transitive closure of $\rightarrow_{G}$ is denoted by $\rightarrow_{G}^{\ast}$.
The \emph{language of $G$} is $L(G) \coloneqq \{ w\in \Sigma^{+} \mid S \rightarrow_{G}^{\ast} w \}$.
Note that with this definition the \emph{empty word}, i.e., the word $\epsilon$ of length $0$, can never be an element of $L(G)$; some authors use a modified definition that also allows rules that derive $\epsilon$, but for our purposes the difference is not essential. 
\begin{example} Let $G\coloneqq (\{S\}, \{a,b\},\{S \rightarrow aSb, S \rightarrow ab\}, S)$. Then it follows by a simple induction that $L(G) = \{a^n b^n \mid n\geq 1 \}$ because every accepting derivation path in $\rightarrow_G$  is of the form
	$
	S \rightarrow_G  \cdots \rightarrow_Ga^{n-1} S b^{n-1} \rightarrow_G a^nb^n.  
	$ 
\end{example}
The idea of the reduction is to compute from a given context-free grammar $G$ a universal Horn sentence which consists of two parts, $\Phi_1$ and $\Phi_2$:  the sentence $\Phi_2$ only depends on $\Sigma$ and 
entails many Horn clauses
witnessing failure of the JEP via Lemma~\ref{charact_jep_horn}; the sentence $\Phi_1$ can be computed efficiently from $G$ and is such that $\Models{\Phi_1}$ is closed under the formation of disjoint unions and 
prevents all the failures of the JEP of $\Models{\Phi_1 \wedge \Phi_2}$ if and only if 
$G$ is universal, i.e., $L(G) = \Sigma^+$.

\begin{theorem} \label{jep_horn_undecidability}  For a given universal Horn sentence $\Phi$ 
	the question whether $\Models{\Phi}$ has the JEP is undecidable even if the signature is limited to at most  binary relation symbols.
\end{theorem}

\begin{proof} 
	The universality problem for context-free grammars is known to be undecidable~\cite{Harrison1978} (Lemma~8.4.2, page 259).
	%   
	%	In this proof, we only consider  context-free  grammars  without production rules of the form   $A \rightarrow w$ such that $w$ contains the start symbol $S$.
	%
	%	This is without loss of generality as every context-free grammar can be transformed into such form in polynomial time~\cite{Harrison1978}.
	%
	Here, we assume that $(A \to A) \in P$ for every  $A\in N $. Note that this assumption does not influence $L(G)$ at all.
	
	\paragraph{Encoding context-free grammars into ages of relational structures.}	
	Let $\tau_1$ be the signature that contains the unary symbols $I$ and $T$,  and the binary relation symbol $R_a$ for every element $a\in N\cup \Sigma$. 
	Let $\Phi_{1}$ be the universal Horn sentence that contains, for every $(A \rightarrow a_1\dots a_n) \in P$, the Horn clause  
	\begin{align}  
		\bigwedge_{i\in [n]} R_{a_i}(x_i,x_{i+1})  & \Rightarrow  R_A(x_1,x_{n+1}), \label{eq:trans-inductive}
	\end{align}
	and additionally the Horn clause  
	\begin{align}  
		I(x_1) \wedge T(x_{2})  \wedge R_{S}(x_1,x_{2})  & \Rightarrow \bot .  \label{eq:trans-s-inductive}
	\end{align} 
	Note that each conjunct  of $\Phi_{1}$ is connected, which means that $\Models{\Phi_{1}}$ has the JEP by Proposition~\ref{connected}. 
	The following correspondence can be shown via a straightforward induction. 
	\begin{myclaim} \label{correspondence} For every $ w=a_1\dots a_n \in (N\cup \Sigma)^{+}$,  we have
		\begin{align}
			A \rightarrow^{\ast}_{G} w \quad \mbox{ if and only if } \quad  \Phi_{1} \models \forall x_{1},\dots ,x_{n+1}    \nonumber \\    \Big(    \bigwedge_{i\in[n]} R_{a_i}(x_i,x_{i+1})   \Rightarrow R_A(x_1,x_{n+1}) \Big).  \label{eq:goal_nobot} \end{align}
		%	For $A=S$, we have 
		%	\begin{align}
		%	S \rightarrow^{\ast}_{G} w \quad \mbox{ if and only if } \quad  \Phi_{1} \models \forall x_{1},\dots ,x_{n+1}    \nonumber \\    \Big(  I(x_1) \wedge T(x_{n+1})  \wedge \bigwedge_{i\in[n]} R_{a_i}(x_i,x_{i+1})   \Rightarrow \bot \Big).  \label{eq:goal} \end{align}
	\end{myclaim} 
	\begin{proof} ``$\Rightarrow$'': Suppose that $A \rightarrow^{\ast}_{G}a_1\dots a_n$ for $A \in N$. 
		Then there is a path in $\rightarrow_G$ from $A$ to $a_1\dots a_n$ of length $\lambda \geq 1$.
		We prove the statements by induction on $\lambda$.
		
		In the \emph{induction base} $\lambda=1$ we have  $(A \rightarrow a)\in P$ in
		which case    	\eqref{eq:goal_nobot} is  a conjunct of $\Phi_1$.
		
		In the \emph{induction step} $\lambda \longrightarrow \lambda+1$, we assume that the claim holds for all paths of length $\leq \lambda$, and that there exists a path of length $\lambda+1$ from $A$ to $a_1\dots a_n$, i.e.,
		there exists $(A \to a_1\dots a_{k-1}Ba_{\ell+1}\dots a_n) \in P$  and a path of length $\lambda$ from $B$ to $a_k \dots a_{\ell}$.
		By the induction hypothesis,  after renaming of variables we  have that 
		\begin{align} 
			\Phi_{1} \models \forall x_{k},\dots ,x_{\ell+1} & \Big(  \bigwedge_{i\in [\ell]\setminus [k-1]}  R_{a_{i}}(x_{i},x_{i+1})  
			\Rightarrow R_B(x_k,x_{\ell+1}) \Big) . \label{eq:deriv1}  
		\end{align}
		By the construction of $\Phi_{1}$, 
		%	if $A \neq S$, 
		after renaming of variables we also have that 
		\begin{align} 
			\Phi_{1} \models \forall x_{1},\dots,x_{k}, x_{\ell+1},\dots, x_{n+1}  \Big(  \bigwedge_{i\in [k-1]} R_{a_i}(x_i,x_{i+1})   \nonumber  \\ \wedge\  R_{B}(x_{k},x_{\ell+1})  \wedge \bigwedge_{i\in [n]\setminus [\ell]} R_{a_i}(x_i,x_{i+1})  
			\Rightarrow R_A(x_1,x_{n+1}) \Big) . \label{eq:deriv2}  
		\end{align}
		We can now apply an SLD derivation step to~\eqref{eq:deriv2} with~\eqref{eq:deriv1} to obtain 
		\eqref{eq:goal_nobot}.
		% 
		%The argument for $A=S$ is analogous. 
		
		``$\Leftarrow$'': Suppose that 	$\Phi_{1} \models \eqref{eq:goal_nobot}$. 
		By Theorem~\ref{SLD-deduction}, there is  an SLD-deduction of \eqref{eq:goal_nobot} from $\Phi_1$.
		If \eqref{eq:goal_nobot} is a tautology, then $n= 1$ and $a_1=A$ in which case the statement is true because $(A \to A) \in P$ for every $A\in N$. 
		Otherwise,  \eqref{eq:goal_nobot} is a weakening of a Horn clause $\psi$ that has an SLD-derivation from $\Phi_1$ modulo renaming variables.
	    Note that the removal of any atom from the premise of \eqref{eq:goal_nobot} would make it disconnected. The same also applies to $\psi$ since its atomic subformulas are among the atomic subformulas of \eqref{eq:goal_nobot}.
			Since every Horn clause in $\Phi_1$ is connected, and SLD-derivations modulo renaming variables clearly preserve connectedness, it cannot be the case that \eqref{eq:goal_nobot} was obtained from $\psi$ by adding an atom to the premise.
			Next, note that $\Phi_1 \wedge \bigwedge_{i\in [n]} R_{a_i}(x_i,x_{i+1})$ is satisfiable for every $a_1\dots a_n \in (N\cup \Sigma)^{+}$.
			Thus, it also cannot be the case that \eqref{eq:goal_nobot} was obtained from $\psi$ by adding an atom to the conclusion. 
		    Hence, we may assume that  \eqref{eq:goal_nobot} and $\psi$ are equal.
		   We prove the claim by induction on the length $\lambda$ of a shortest possible SLD-derivation for $\psi$. 
		
		In the \emph{base case} $\lambda=0$,   
		$\psi$ 
		must be a conjunct of $\Phi_1$.
		%
		%Since each $a_i$ is a terminal symbol, $\psi$ must be of the form \eqref{eq:terminal}.
		%
		By the construction of $\Phi_{1}$, we get that $(A \to a_1\dots a_n) \in P$ and thus $A \rightarrow^{*}_{G}a_1\dots a_n$.  
		
		In the \emph{induction step} $\lambda \longrightarrow \lambda+1$, we assume that the claim holds if $\psi$ has an SLD-derivation of length $\leq \lambda$.
		Suppose that $\psi$ requires an SLD-derivation of length $\lambda+1$. 
		By the construction of  $\Phi_{1}$, there must exist $(B,a_k\dots a_{\ell} ) \in P$ such that $\Phi_{1}$ contains a conjunct of the form
		\eqref{eq:trans-inductive}	that is used in the last step in a shortest possible SLD-derivation of $\psi$. 
		Moreover, there  exists  an SLD-derivation of
		\begin{align}
			\Big(	\bigwedge_{i\in [k-1]} R_{a_i}(x_i,x_{i+1})  \Big)   \wedge\  R_{B}(x_{k},x_{\ell+1})  \nonumber \\ \wedge \bigwedge_{i\in [n]\setminus [\ell]} R_{a_i}(x_i,x_{i+1})  
			\Rightarrow R_A(x_1,x_{n+1})  \label{eq:toshow3} 
		\end{align} 
		from $\Phi_{1}$ of length $\leq \lambda$.
		By the induction hypothesis, \eqref{eq:toshow3} is equivalent to $A \rightarrow^{\ast}_{G} a_1 \dots a_{k-1} B a_{\ell+1}\dots a_{n-1}$. Therefore, $A \rightarrow^{\ast}_{G}  a_1,\dots, a_n$. %The argument for $A=S$ is analogous. 
	\end{proof}

	\paragraph{Creating candidates for failure of the JEP.}
	Let $\tau_{2}$ be the signature which contains  all symbols from $\tau_1$ except for the ones coming from $N$ and additionally the unary symbol $U$ and the binary symbol $Q$.
	The sentence $\Phi_2$ consists of the following Horn clauses for every $a \in \Sigma$:
	\begin{align}   
		U(y)  \wedge I(x_1)  \Rightarrow\ & Q(y,x_1) \label{eq:deriving_Q_from_S} \\ 
		U(y) \wedge Q(y,x_1)  \wedge R_{a}(x_1,x_{2})    \Rightarrow\ & Q(y,x_{2}) \label{eq:deriving_Q_from_S_and_Q}  \\ 
		U(y) \wedge Q(y,x_{1}) \wedge  R_{a}(x_{1},x_{2})  \wedge T(x_{2})    \Rightarrow \ &  \bot  \label{eq:deriving_bot_from_T_and_Q}
	\end{align} 
	
	\begin{figure}[t]
		\begin{center}
			\begin{tikzpicture}[scale=0.5]  
				\foreach \phi in {1,...,6}{
					\node[minimum size=7pt,inner sep=1pt,outer sep=1pt] (x\phi) at (360/8 * \the\numexpr\phi+0.5: 2.65cm) {$x_\phi$};} 
				\foreach \phi in {1,...,5}{
					\node[minimum size=7pt,inner sep=1pt,outer sep=1pt] (r\phi) at (360/8 * \the\numexpr\phi+1: 3.75cm) {$R_{a_\the\numexpr\phi}$};} 
				\node[minimum size=7pt,inner sep=1pt,outer sep=1pt] (y) at (360/8 * 0: 1.025cm) {$y$}; 
				\draw [opacity=1, draw=blue,-,line width=0.7mm,double distance = 0.55cm,line cap=round] plot [smooth,tension=0 ] coordinates {  (x1) (x1) }; 
				\draw [opacity=1, draw=red,-,line width=0.7mm,double distance = 0.55cm,line cap=round] plot [smooth,tension=0 ] coordinates {  (x6) (x6) }; 
				\draw [opacity=1, draw=darkgreen,-,line width=0.7mm,double distance = 0.55cm,line cap=round] plot [smooth,tension=0 ] coordinates {  (y) (y) }; 
				\foreach \phi in {1,...,5}{
					\draw [opacity=1, draw=black,-,line width=0.3mm,double distance = 0.5cm,line cap=round] plot [smooth,tension=0 ] coordinates {  (x\phi) (x\the\numexpr\phi+1) }; } 
				\foreach \phi in {1,...,6}{
					\node[ minimum size=7pt,inner sep=1pt,outer sep=1pt] (x\phi) at (360/8 * \the\numexpr\phi+0.5: 2.65cm) {$x_\phi$};} 
				\node[ minimum size=7pt,inner sep=1pt,outer sep=1pt] (y) at (360/8 * 0: 1.025cm) {$y$}; 
				\node[ minimum size=7pt,color=blue, inner sep=1pt,outer sep=1pt] (y) at (360/8 * 1.2: 3.6cm) {$I$};
				\node[ minimum size=7pt,color=darkgreen,inner sep=1pt,outer sep=1pt] (y) at (360/8 * 0: 2.2cm) {$U$};
				\node[ minimum size=7pt,color=red,inner sep=1pt,outer sep=1pt] (y) at (360/8 * -1.2: 3.6cm) {$T$}; 
			\end{tikzpicture} 
			\hspace{3em}\raisebox{6em}{{\Large $\overset{\Phi_2\,}{\Longrightarrow}$}} \hspace{2em}
			\begin{tikzpicture}[scale=0.5]  
				\foreach \phi in {1,...,6}{
					\node[minimum size=7pt,inner sep=1pt,outer sep=1pt] (x\phi) at (360/8 * \the\numexpr\phi+0.5: 2.65cm) {$x_\phi$};} 
				\foreach \phi in {1,...,5}{
					\node[minimum size=7pt,inner sep=1pt,outer sep=1pt] (r\phi) at (360/8 * \the\numexpr\phi+1: 3.75cm) {$R_{a_\the\numexpr\phi}$};} 
				\node[minimum size=7pt,inner sep=1pt,outer sep=1pt] (y) at (360/8 * 0: 1.025cm) {$y$}; 
				\draw [opacity=1, draw=blue,-,line width=0.7mm,double distance = 0.55cm,line cap=round] plot [smooth,tension=0 ] coordinates {  (x1) (x1) }; 
				\draw [opacity=1, draw=red,-,line width=0.7mm,double distance = 0.55cm,line cap=round] plot [smooth,tension=0 ] coordinates {  (x6) (x6) }; 
				\draw [opacity=1, draw=darkgreen,-,line width=0.7mm,double distance = 0.55cm,line cap=round] plot [smooth,tension=0 ] coordinates {  (y) (y) }; 
				\foreach \phi in {1,...,5}{
					\draw [opacity=1, draw=black,-,line width=0.3mm,double distance = 0.5cm,line cap=round] plot [smooth,tension=0 ] coordinates {  (x\phi) (x\the\numexpr\phi+1) }; } 
				\foreach \psi in {1,...,5}{ 
					\draw [opacity=1, draw=purple,-,line width=0.3mm,double distance = 0.55cm,line cap=round] plot [smooth,tension=2 ] coordinates {  (y) (x\psi) }; }  
				\foreach \phi in {1,...,6}{
					\node[ minimum size=7pt,inner sep=1pt,outer sep=1pt] (x\phi) at (360/8 * \the\numexpr\phi+0.5: 2.65cm) {$x_\phi$};} 
				\foreach \phi in {2,...,5}{
					\node[ color=purple, minimum size=7pt,inner sep=1pt,outer sep=1pt] (x\phi) at (360/6.75 * \the\numexpr\phi-1.9 : 1.5cm) {$Q$};} 
				\node[ color=purple, minimum size=7pt,inner sep=1pt,outer sep=1pt] (x1) at (360/6.75 * 1.05 : 1.8cm) {$Q$}; 
				\node[ minimum size=7pt,inner sep=1pt,outer sep=1pt] (y) at (360/8 * 0: 1.025cm) {$y$}; 
				\node[ minimum size=7pt,inner sep=1pt,color=blue,outer sep=1pt] (y) at (360/8 * 1.2: 3.6cm) {$I$};
				\node[ minimum size=7pt,inner sep=1pt,color=darkgreen,outer sep=1pt] (y) at (360/8 * 0: 2.2cm) {$U$};
				\node[ minimum size=7pt,inner sep=1pt,color=red,outer sep=1pt] (y) at (360/8 * -1.2: 3.6cm) {$T$}; 
			\end{tikzpicture} 
		\end{center}
		\caption{An illustration of the situation in Claim~\ref{correspondence_mod} for $n=6$.}
		\label{fig:failures}
	\end{figure}
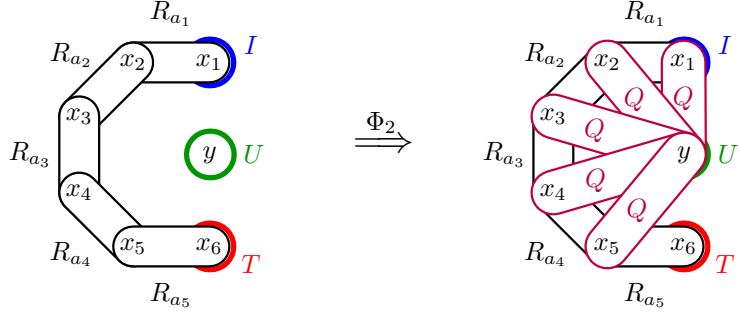
	
	The proof of the following claim is left to the reader (a straightforward consequence of Theorem~\ref{SLD-deduction}), see Figure~\ref{fig:failures} for an illustration of Claim~\ref{correspondence_mod}. 
	\begin{myclaim}  \label{correspondence_mod} 
		Let $\phi(\boldf{x})$ be a conjunction of $(\tau_2\setminus \{Q\})$-atoms.
		Then  $$\Phi_2 \models \forall \boldf{x} \big( \phi(\boldf{x}) \Rightarrow \bot \big)$$ if and only if there exists $a_1\dots a_n \in \Sigma^{+}$ such that $\phi$ has a subformula of the form 
		\begin{align}  
			U(y) \wedge I(x_1)\wedge  T(x_{n+1}) \wedge 
			\bigwedge_{i\in [n]}  R_{a_{i}}(x_i,x_{i+1})     \label{eq:goal_mechanism} 
		\end{align} 
	where the variables need not be distinct.
	\end{myclaim}

	Now we are ready for the main part of the proof. We set $\Phi\coloneqq \Phi_1\wedge \Phi_2$. 
	
	\begin{myclaim}  $\Models{\Phi}$  has the JEP if and only if  $L(G) = \Sigma^{+}$.
		\label{claim:final}
	\end{myclaim} 
	\begin{proof} 	``$\Rightarrow$'': Suppose that $\Models{\Phi}$ has the JEP. 
		Let $a_1\dots a_{n}\in \Sigma^{+}$.  
		Consider the formulas  $\phi_1(y)$ and $\phi_2(x_1,\dots,x_{n+1})$ given by 
		\begin{align*} 
			\phi_1 \coloneqq \ & U(y) \\
			\text{and} \quad \phi_2 \coloneqq \ &  I(x_1)\wedge  T(x_{n+1}) \wedge 
			\bigwedge_{i\in [n]}  R_{a_{i}}(x_i,x_{i+1}) .
		\end{align*} 
		By Claim~\ref{correspondence_mod}, we have 
		\begin{align}  \Phi_2 \models \forall x_1,\dots,x_{n+1},y \big(     \phi_1 \wedge \phi_2  \Rightarrow \bot \big). \label{eq:assumption} 
		\end{align}  
		Since $\Models{\Phi}$ has the JEP, by Lemma~\ref{charact_jep_horn}\eqref{eq:lemma1_item2}, we have that $\Phi \wedge \phi_1$ or $\Phi \wedge \phi_2$ is not satisfiable.

		Note that every Horn clause in $\Phi_1$ contains an $R_a$-atom for some $a\in N\cup \Sigma$ in its premise, and $\phi_1$ contains none.  
		Also note that every Horn clause in $\Phi_2$ contains an $I$-atom or an $R_a$-atom for some $a\in N\cup \Sigma$ in its premise, and $\phi_1$ contains none. 
		Hence, the $\tau$-structure with domain $\{y\}$ and whose relations are described by the atomic formula $U(y)$ is closed under application of Horn clauses from $\Phi$, i.e., $\Phi \wedge \phi_1$ is satisfiable.
		Since one of the formulas $\Phi \wedge \phi_1$ and $\Phi \wedge \phi_2$ is not satisfiable, we must have
		\begin{align} 
		\Phi \models \forall x_1,\dots, x_{n+1} \big( \phi_2 \Rightarrow \bot   \big). \label{eq:pruning1}
\end{align}
		
		Next, note that every Horn clause in $\Phi_2$ contains a $U$-atom in its premise, $\phi_2$ contains no $U$-atom, and no Horn clause in $\Phi_1$ contains a $U$-atom in its conclusion.
	    We claim that then
	    \begin{align} 
	    \Phi_1 \models \forall x_1,\dots, x_{n+1} \big( \phi_2 \Rightarrow \bot   \big). \label{eq:pruning2}
	\end{align}
	    Suppose, on the contrary, that $\Phi_1 \wedge \phi_2$ is satisfiable. Let $\struct{A}$ be any finite model of $\Phi_1 \wedge \phi_2$.
	    We assume that $U^{\struct{A}}=\emptyset$, otherwise we remove all elements from $U^{\struct{A}}$.
	    Then $\struct{A}$ still satisfies $\phi_2$ because $\phi_2$ does not contain any $U$-atoms, and $\struct{A}$  remains closed under application of Horn clauses from $\Phi_1$ because no Horn clause in $\Phi_1$ contains a $U$-atom in its conclusion.
	    But then $\struct{A}$ is also closed under application of Horn clauses from $\Phi_2$ because every Horn clause in $\Phi_2$ contains a $U$-atom in its premise.
	    We conclude that $\struct{A} \models \Phi \wedge \phi_2$, a contradiction to \eqref{eq:pruning1}.
	    Hence \eqref{eq:pruning2} holds.   

		Finally, note that \eqref{eq:trans-s-inductive} is the only Horn clause in $\Phi_1$ which contains  an $I$- or $T$-atom in its premise, and that no Horn clause in $\Phi_1$ contains an $I$- or $T$-atom in its conclusion.
		 We claim that then
        \begin{align} 
        		\Phi_1 \models \forall x_1,\dots, x_{n+1} \Big( 	\bigwedge_{i\in [n]}  R_{a_{i}}(x_i,x_{i+1})  \Rightarrow R_S(x_1,x_{n+1})   \Big). \label{eq:pruning3}
        \end{align}
        Suppose, on the contrary, that \eqref{eq:pruning3} does not hold.
        Then there exists a finite $\tau$-structure $\struct{A}$ satisfying $\Phi_1 \wedge \bigwedge_{i\in [n]}  R_{a_{i}}(x_i,x_{i+1})\wedge \neg R_S(x_1,x_{n+1})$.
        We may assume that $A= \{x_1,\dots, x_{n+1}\}$ because the models of a universal sentence are always closed under taking substructures.
        We may also assume that $I^{\struct{A}}=T^{\struct{A}}=\emptyset$, otherwise we remove all elements from these relations.
        Then $\struct{A}$ remains closed under application of Horn clauses from $\Phi_1$ because no Horn clause in $\Phi_1$ contains an $I$- or $T$-atom in its conclusion.
        Now consider the structure $\struct{A}'$ obtained from $\struct{A}$ by adding $x_1$ to $I^{\struct{A}}$ and $x_{n+1}$ to $T^{\struct{A}}$.
        Then $\struct{A}'$ is closed under application of Horn clauses from $\Phi_1$ because  \eqref{eq:trans-s-inductive} is the only Horn clause in $\Phi_1$ that has an $I$- or $T$-atom in its premise and it cannot be applied to $\struct{A}'$ since $(x_1,x_{n+1}) \notin R_S^{\struct{A}'}$.
        But then $\struct{A}'\models \Phi_1 \wedge \phi_2$, a contradiction to \eqref{eq:pruning2}.
        Hence \eqref{eq:pruning3} holds. 
		
		Now it follows from Claim~\ref{correspondence} and \eqref{eq:pruning3} that $a_1\dots a_n \in L(G)$ and we are done.
		
		``$\Leftarrow$'':  
		We prove the contrapositive and assume that $\Models{\Phi}$ does not have the JEP.
		Then there exists a counterexample to   Lemma~\ref{charact_jep_horn}\eqref{eq:lemma1_item2}, i.e., 
		there exists a Horn clause $\psi$ of the following form
		$$   \phi_1(\boldf{x}_1) \wedge \phi_2(\boldf{x}_2) \Rightarrow \chi$$
		where $\chi$ is either $\bot$ or an atomic $\tau$-formula 
		with free variables among $\boldf{x}_1$  
		such that  
		\begin{align}   
			&\Phi \models   \forall \boldf{x}_1,\boldf{x}_2\big(  \phi_1 \wedge \phi_2 \Rightarrow \chi \big), \label{eq:badrule}  \\ 
			&\Phi \centernot{\models}   \forall \boldf{x}_1\big( \phi_1 \Rightarrow \chi\big), \label{eq:violation}
		\end{align}
		and, for both $i \in \{1,2\}$,  
		\begin{align}  
			\Phi \centernot{\models}	\forall \boldf{x}_i \big( \phi_i   \Rightarrow \bot \big).  \label{eq:ass}
		\end{align} 
		We choose $\psi$ minimal with respect to  the number of its atomic subformulas. 
		
		Our proof strategy is as follows. 
		First we show that $\psi$ encodes a single word $w\in \Sigma^+$ in the sense of Claim~\ref{correspondence_mod}.
		Then we show that the word $w$ may not be contained in $L(G)$, because otherwise a part of the counterexample would encode $w$ in the sense of Claim~\ref{correspondence} which would lead to a contradiction.

		\begin{observation} \label{obs:B}  $\psi$ has an SLD-deduction from $\Phi_2$, and only contains symbols from $\tau_2$.
		\end{observation}   
		\begin{proof}[of Observation~\ref{obs:B}]   
		 %
          %Note that $Q\notin \tau_1$ and $R_A \notin \tau_2$ for every $A\in N$. 
          First, we claim that $\Phi_1 \vdash\psi$ or $\Phi_2 \vdash \psi$. 
		 By  Theorem~\ref{SLD-deduction}, we have $\Phi \vdash \psi$.  
		 Note that $\chi(\boldf{x}_1)$ cannot be a subformula of 
		 $   \phi_1(\boldf{x}_1)$, by \eqref{eq:violation}.  
		 Also note that $\chi(\boldf{x}_1)$ cannot be a subformula of $\phi_2(\boldf{x}_2)$ as these two formulas have no common variables.
		 Hence, $\chi(\boldf{x}_1)$ is not a subformula of $  \phi_1(\boldf{x}_1) \wedge \phi_2(\boldf{x}_2)$, i.e., $\psi$ is not a tautology.
		 Let $\psi'$ be a Horn clause such that $\psi$ is a weakening of $\psi'$ and $\psi'$ has an SLD-derivation $\psi'_{0},\dots,\psi'_{k}=\psi'$  from $\Phi$.
		 Note that the Horn clauses in $\Phi$ have the property that, depending on whether they come from $\Phi_1$ or from $\Phi_2$, they either contain no $Q$-atoms or no $R_A$-atoms for $A\in N$.
		 This applies in particular to $\psi'_{0}$ which is a conjunct from $\Phi$.
		 Since the conclusion of each Horn clause in $\Phi_1$ is an $R_A$-atom for $A\in N$ and the conclusion of each Horn clause in $\Phi_2$ is a $Q$-atom, the property of $\psi'_{0}$ from above propagates inductively to every $\psi'_i$ for $i\in [k]$.
		 But this means that $\psi'$ has an SLD-derivation from $\Phi_1$ or from $\Phi_2$.
		 Hence, $\psi$ has an SLD-deduction from $\Phi_1$ or from $\Phi_2$, which concludes the claim.

			Next, we claim that $\Phi_2 \vdash \psi$. Suppose, on the contrary, that $\Phi_1 \vdash \psi$.
			Let $\phi'_1$ and $\phi'_2$ be the formulas obtained from $\phi_1$ and $\phi_2$, respectively, by removing all $Q$-atoms. 
			Since $\Phi_1 \vdash \psi$, the SLD-derivation sequence $\psi'_{0},\dots,\psi'_{k}$ from the paragraph above contains no $Q$-atoms.
			Thus, all $Q$-atoms occurring in $\psi$ come from the weakening step,  which means that  
			\begin{align} 
				\Phi_{1} \vdash \forall \boldf{x}_1,\boldf{x}_2\big(\phi_1' \wedge \phi_2' \Rightarrow \chi \big). \label{eq:G2_reduct}
			\end{align}
			Since $\psi'_{0}$ is a Horn clause from $\Phi_1$, it follows from the minimality assumption for $\psi$ that  either $\chi$ is an $R_{A}$-atom for some $A\in N$, or $\chi$ is $\bot$.
			In both cases, \eqref{eq:ass}, \eqref{eq:G2_reduct}, and \eqref{eq:violation} witness that $\Models{\Phi_1}$ does not have JEP through an application of Lemma~\ref{charact_jep_horn}.
			But this is in contradiction to Proposition~\ref{connected}. Thus, $\Phi_1 \vdash \psi$ does not hold, and $\Phi_2 \vdash \psi$ holds instead. This concludes the claim.

			Since $\Phi_2 \vdash \psi$, the premise $\phi_1 \wedge \phi_2$ of $\psi$ can only contain symbols from $\tau_2$,  otherwise we could remove all $(\tau_1 \setminus \tau_2)$-atoms and get a contradiction to the minimality of $\psi$.   
			Since $\psi'_{0}$ is a Horn clause from $\Phi_2$, it also follows from the minimality assumption for $\psi$ that  either $\chi$ is an $Q$-atom, or $\chi$ is $\bot$. 
			Thus, $\psi$ only contains symbols from $\tau_2$.
		\end{proof}
		\begin{observation}  \label{obs:C} $  \phi_1  \wedge \phi_2 $ does not contain any $Q$-atoms, and $\chi$ equals $\bot$. 
			%
			%\red{Moreover, %we can always choose $\chi$ as $ \bot$ independently of $\phi(\boldf{x})\wedge \phi_1(\boldf{x},y_1) \wedge \phi_2(\boldf{x},y_2)$. 
			% \eqref{eq:badrule} and~\eqref{eq:violation} also hold if $\chi$ is replaced by $\bot$. }
			%    	 Moreover, we  may assume that, for every choice of a minimal counterexample $\phi(\boldf{x})\wedge \phi_1(\boldf{x},y_1)\wedge \phi_2(\boldf{x},y_2)$, $\chi$ equals $\bot$ in  \eqref{eq:badrule}  and \eqref{eq:violation}.}
		\end{observation}   
		\begin{proof}[of Observation~\ref{obs:C}] 
			By Observation~\ref{obs:B}, $\psi$ has an SLD-deduction from $\Phi_2$.
			Recall from the proof of Observation~\ref{obs:B} that $\psi$ cannot be a tautology.
			%   By the minimality of our counterexample, no $Q$-atoms may be present in the premise of $\psi$ due to the assumption of minimality of our counterexample as their removal  preserves
			%   \eqref{eq:ass}, \eqref{eq:badrule}, and \eqref{eq:violation}.
			By the minimality of $\psi$, we %therefore 
			may
			assume that there exists an SLD-derivation of $\psi$ from $\Phi_{2}$.  
			Consider any SLD-derivation $\psi_{0},\dots, \psi_{k}$ of $\psi$ from $\Phi_2$.
			Note that, by the construction of $\Phi_2$, for every $i\in[k]$, if there exists a variable $y$ in $\psi_{i-1}$ such that
			\begin{enumerate}%[label=\roman*.,ref=\roman*]
				\item \label{eq:item1} every $Q$-atom contains   $y$ in its first argument,  
				\item \label{eq:item2} every $U$-atom contains   $y$ in its only argument,
			\end{enumerate}  
			then $\psi_{i}$ also satisfies  item~\ref{eq:item1} and  item~\ref{eq:item2} for the same variable $y$.
			Since every possible choice of  $\psi_0$ from $\Phi_2$ initially satisfies these two conditions, it follows via induction that item~\ref{eq:item1} and item~\ref{eq:item2} must hold for $\psi = \psi_{k}$ for some $y$. 
			Also note that \eqref{eq:deriving_Q_from_S} is the only Horn clause in $\Phi_2$ that is not connected, 	but the lack of connectivity is only because the variable $y$ satisfying item~\ref{eq:item1} and item~\ref{eq:item2} is isolated from the remaining variables.
			It follows by induction that this is also true for $\psi$.
				%the only reason why $\psi=\psi_k$ might not be connected is that the variable $y$ that satisfies item~\ref{eq:item1} and item~\ref{eq:item2} for $\psi$ might be isolated from the remaining variables.

		 We claim that $\psi_0$ is of the form \eqref{eq:deriving_bot_from_T_and_Q}. 
				Suppose, on the contrary, that $\psi_0$ is of the form \eqref{eq:deriving_Q_from_S_and_Q} or \eqref{eq:deriving_Q_from_S}. Then the conclusion of $\psi$ is a $Q$-atom.  
				By our assumption, the conclusion of $\psi$ may only contain variables from $\boldf{x}_1$.
				Thus, also the variable $y$ that satisfies item~\ref{eq:item1} and item~\ref{eq:item2} for $\psi$ is contained in $\boldf{x}_1$.
				Since the second variable in the conclusion of $\psi$ is connected to all remaining variables and the variables of $\phi_1$ and $\phi_2$ are disjoint, $\phi_2$ must be the empty conjunction.
				This leads to a contradiction to \eqref{eq:violation}.
				Thus our claim holds.
				The claim implies that $\chi$ equals $\bot$.

				Since $y$ satisfies item~\ref{eq:item1} for $\psi$, if  $\psi$ contains any $Q$-atom in the premise, then $\psi$ is connected.
				But then, since the variables of $\phi_1$ and $\phi_2$ are disjoint while $\psi$ is connected, either $\phi_1$ or $\phi_2$ must be the empty conjunction.
				Since $\chi$ equals $\bot$, this leads to a contradiction to \eqref{eq:ass}.
			Thus,  $\psi$ does not contain any $Q$-atoms at all.
		\end{proof}	
	 
		As a consequence of   
		Observation~\ref{obs:B} and Observation~\ref{obs:C}
		we have that $\Phi_2 \models \psi$
		where $\psi$ is of the form $ \phi_1 \wedge \phi_2 \Rightarrow \bot$
		and $ \phi_1,\phi_2$ are conjunctions of atomic $\tau_2 \setminus \{Q\}$-formulas. 
		Therefore, Claim~\ref{correspondence_mod}  implies that there exists $a_1 \dots  a_n \in \Sigma^{+}$ such that $ \phi_1 \wedge \phi_2$ is of the form 
		\begin{align}
			U(y) \wedge I(x_1)\wedge  T(x_{n+1}) \wedge 
			\bigwedge_{i\in [n]}  R_{a_{i}}(x_i,x_{i+1})     \label{eq:final_form}
		\end{align}  
		where the variables need not all be distinct.  
		More specifically,  
		\begin{itemize}
			\item $\phi_1(\boldf{x}_1)$ equals $U(y)$, and
			\item $\phi_2(\boldf{x}_2)$ equals $I(x_1)\wedge  T(x_{n+1}) \wedge 
			\bigwedge_{i\in [n]}  R_{a_{i}}(x_i,x_{i+1})  $.
		\end{itemize} 
		Note that, if $L(G)=\Sigma^{+}$, 
		then Claim~\ref{correspondence} together with \eqref{eq:trans-s-inductive} implies that
		\begin{align} \Phi_{1} \models \forall x_1,\dots, x_{n+1}\big(  \phi_2   \Rightarrow \bot \big). \label{eq:second_grammar}
		\end{align} 
		If some variables among  $x_1,\dots, x_{n+1}$ are identified in \eqref{eq:final_form}, then we still have \eqref{eq:second_grammar} even if we perform the same identification of variables.
		But then we get a contradiction to \eqref{eq:ass}.
		Thus,  $L(G) \neq \Sigma^{+}$, which concludes the proof of Claim~\ref{claim:final}.   
	\end{proof}	 
	
	We have thus found a reduction from the undecidable universality problem for $G$ to the decidability problem of the JEP for $\Models{\Phi}$; note that $\Phi$ is universal Horn and can be computed from $G$ in polynomial time. 
\end{proof}

\section{Open Problems}
An important open question about universal sentences over relational signatures is the decidability of the \emph{amalgamation property (AP)}, 
which is a strong form of the joint embedding property, and which is of fundamental importance in constraint satisfaction~\cite{Book}. 
Unlike the JEP,  the AP is decidable if all relation symbols in the signature are at most binary (see, e.g.,~\cite{KnightLachlan,BodirskyKnaeuerStarke}). 
For general relational signatures, we ask the following question.

\begin{question}
	Is it decidable whether the class of finite models of a given universal Horn sentence $\Phi$ has the amalgamation property? 
\end{question}

%Note that if the class of finite models of $\Phi$ does have amalgamation, then we can compute an amalgam syntactically by computing implied atomic formulas just as in the proof of Lemma~\ref{charact_jep_horn}. 

\section*{Acknowledgments}{We thank the anonymous referees for their comments, in particular for one that was leading to the correct formulation of Proposition~\ref{prop:equivalence_horn}.}

\nocite{*}
%\bibliographystyle{abbrvnat}
% use the following instead if you encounter problems 
\bibliographystyle{alpha}
\bibliography{jep_horn}
\label{sec:biblio}

\end{document}